\documentclass[lettersize,journal]{IEEEtran}

\usepackage{amsmath, amsfonts, amssymb, amsthm} 
\usepackage{algorithm, algpseudocode} 
\usepackage{textcomp}
\usepackage{stfloats}
\usepackage{url}
\usepackage{verbatim}
\usepackage{graphicx}
\usepackage[noadjust]{cite}
\usepackage{multirow}
\usepackage[dvipsnames]{xcolor}
\usepackage{threeparttable}
\usepackage{setspace}
\usepackage{pgfplots}
\usepgfplotslibrary{fillbetween}
\pgfplotsset{compat=1.18}
\usepackage{diagbox}
\usepackage{array, makecell}
\usepackage{hyperref}
\usepackage{enumerate, xspace, bbold, bm} 
\usepackage{enumitem} 
\usepackage{float} 
\usepackage{tikz} 
\usepackage{subcaption}

\newtheorem{theorem}{Theorem}
\newtheorem{proposition}[theorem]{Proposition}
\newtheorem{corollary}[theorem]{Corollary}

\newtheorem{definition}[theorem]{Definition}
\newtheorem{example}[theorem]{Example}

\newcommand{\RR}{\ensuremath{\mathbb{R}}\xspace}

\newcommand{\NP}{{\bf(NP)}\xspace}

\newcommand{\OP}{{\bf(OP)}\xspace}

\newcommand{\TC}{{\bf(TC)}\xspace}

\newcommand{\MTC}{{\bf(MTC)}\xspace}

\newcommand{\TP}{\ensuremath{T_{\bm{P}}}\xspace}

\newcommand{\TM}{\ensuremath{T_{\bm{M}}}\xspace}

\newcommand{\ILK}{\ensuremath{I_{\bm{LK}}}\xspace}
\newcommand{\TLK}{\ensuremath{T_{\bm{LK}}}\xspace}
\newcommand{\IGD}{\ensuremath{I_{\bm{GD}}}\xspace}
\newcommand{\IGG}{\ensuremath{I_{\bm{GG}}}\xspace}

\newcommand{\IT}{\ensuremath{I^T_{\varphi,f,g}}\xspace}

\newcommand{\New}[1]{\textcolor{black}{#1}}

\begin{document}

\title{Fuzzy Implicative Rules: A Unified Approach}
\author{Raquel Fernandez-Peralta\thanks{R. Fernandez-Peralta works in the Mathematical Institute of the
				Slovak Academy of Sciences (SAS), Bratislava, Slovakia. e-mail:
				raquel.fernandez@mat.savba.sk.
				This work was funded by the EU NextGenerationEU through the Recovery
				and Resilience Plan for Slovakia under the project No. 09I03-03-V04-00557.}}

\maketitle

\begin{abstract}
Rule mining algorithms are one of the fundamental techniques in data mining for disclosing significant patterns in terms of linguistic rules expressed in natural language. In this paper, we revisit the concept of fuzzy implicative rule to provide a solid theoretical framework for any fuzzy rule mining algorithm interested in capturing patterns in terms of logical conditionals rather than the co-occurrence of antecedent and consequent. In particular, we study which properties should satisfy the fuzzy operators to ensure a coherent behavior of different quality measures. As a consequence of this study, we introduce a new property of fuzzy implication functions related to a monotone behavior of the generalized modus ponens for which we provide different valid solutions. Also, we prove that our modeling generalizes others if an adequate choice of the fuzzy implication function is made, so it can be seen as an unifying framework. We test the applicability and relevance of our framework for different real datasets and fuzzy operators.
\end{abstract}

\begin{IEEEkeywords}
Knowledge discovery in databases (KDD), 
Rule-based models, 
Implicative rules,
Association rules, 
Fuzzy logic, 
Fuzzy implication function.
\end{IEEEkeywords}

\newtheorem{interpretation}{Interpretation}
\newtheorem{explanation}{Explanation}


\section{Introduction}\label{sec:introduction}

\IEEEPARstart{D}{ata} mining or Knowledge Discovery in Databases (KDD) is defined as the automatic extraction of patterns representing knowledge implicitly stored or captured in data \cite{Fayyad1996}. Up to now, a wide variety of data mining techniques have been introduced and developed \cite{Han2012}. These techniques are usually divided into two types: exploratory and predictive. Exploratory data analysis focuses on searching relations between objects of a dataset (clustering, association rules, linguistic summaries...), whereas predictive data analysis aims to extract knowledge from discovered data with the intent to predict or classify unknown examples (classification, regression...). Within this context, rule-based algorithms have been one of the top choices for knowledge extraction because of their human-understandable output. Furthermore, with the emergence of the research area of eXplainable Artificial Intelligence (XAI) \cite{Barredo2020}, they have become even more popular since they are highly interpretable models, which are an appealing solution to provide an easily understandable representation of complex black-box models \cite{Macha2022,fumanalex2024}.

In rule-based models, the output is given in terms of rules which are usually represented as $A \Rightarrow C$ in which $A$ is called the antecedent and $C$ is called the consequent. Although in these techniques the output has the same representation, depending on the knowledge to be captured we can find several rule mining techniques, just to mention some of them: (i) \textit{association rules}: it is interested in capturing situations in which if $A$ is satisfied, then $C$ is likely to occur also \cite{Agrawal1996}; (ii) \textit{classification rules}: it aims to obtain a set of rules which effectively classify items into predefined classes \cite{Mendel2024}; (iii) \textit{sequential rules}: it searches for rules that indicate that if some event(s) occurred, some other event(s) are likely to occur \cite{Das1998}; (iv) \textit{subgroup discovery}: it aims to obtain rules with the most unusual statistical characteristics with respect to a property of interest \cite{Atzmueller2015}, etc.

Despite the representation of the rules as logical conditionals, they are usually interpreted and evaluated as the co-occurrence of antecedent and consequent. Although this assumption is quite conventional in the literature, it has also received some criticism \cite{Brin1997,Hullermeier2001}. Indeed, the discrepancy between the linguistic representation of the rules and their underlying mathematical modeling can be confusing for the expert. For instance, if one considers the rule \textit{$\text{Smoking} \Rightarrow \text{Respiratory Problems}$} it might seem unnatural that the support of this rule is not affected by its direction, because there exists a lot of respiratory problems that are not associated with smoking. An alternative to overcome this drawback is to consider a more logic-oriented approach in which the conditional in the rule is interpreted as a logical implication, in this case normally the term ``implicative rules'' is used.  In the crisp setting, the authors in \cite{Brin1997} proposed to interpret the logical conditional through a measure called ``conviction'' as an alternative to confidence which is based on conditional probability. However, in the fuzzy setting, fuzzy implication functions can be directly used. Indeed, in \cite{Hullermeier2001} the author introduced implication-based fuzzy association rules by using a t-norm for modeling conjunction, a fuzzy implication function to model the logical conditional, and the generalized modus ponens as inference schema. Nonetheless, no further development of this framework or a thorough theoretical study of the proposal was made.

Further, it is well known that the theoretical development of fuzzy operators has skyrocketed in the last two decades. Indeed, the introduction of new fuzzy operators and the study of their properties is nowadays so vast that there exist many different families from among which we can choose. For instance, only in the case of fuzzy implication functions it has been recently gathered that more than 100 families have been introduced in the literature \cite{Fernandez-Peralta2024}. This fact, although clearly positive as it indicates great flexibility, also carries with it the intrinsic problem of deciding which operator is the most suitable for each application. In this sense,  Zimmerman defines eight different criteria for an adequate selection of the involved aggregation operators in a specific model \cite{Zimmermann1991}, although authors like Mendel point out that these criteria are rather subjective to be successfully implemented in engineering applications \cite{Mendel1995}. Thus, the adequate selection of the set of operators is still a hot topic of discussion nowadays.

In this paper, we revisit the definition given in \cite{Hullermeier2001} for modeling fuzzy implicative rules and we contextualize it with respect to the advances in fuzzy operators made in recent years. We generalize the four most used quality measures (coverage, support, confidence, and weighted relative accuracy) to this new setting. Then, we study which properties the pair of operators $(T,I)$ should satisfy to obtain a coherent behavior of these measures. One of the key points of this study is to prove that for the support to be monotone with respect to the refinements of a rule, we have to introduce an additional property of fuzzy implication functions related to the generalized modus ponens. To our knowledge, this property had not been previously introduced in the literature, so we have studied it in the paper, and we have found several solutions. \New{Also, we prove that not every choice of pair of operators $(T,I)$ models directional rules; in fact, we show that the most common fuzzy implication functions are equivalent to the use of t-norms in this setting, so we propose different pairs of operators that yield information beyond what t-norms can provide.}

Further, we provide an open-source implementation in Python to mine fuzzy implicative associative rules according to the proposed framework. The implementation allows a custom selection of the fuzzy operators and fuzzy partitions. Finally, we have tested our algorithm in different publicly available databases and admissible pairs of operators. From the results, we have exposed that our perspective provides valuable knowledge which is different to that obtained by other approaches.

The structure of the paper is as follows. First, in Section \ref{sec:preliminaries} we include basic results and definitions. In Section \ref{sec:modeling} we introduce fuzzy rules modeled as conditionals, we define different quality measures and we study which properties should the fuzzy operators satisfy. In Section \ref{sec:operator} we provide solutions for the joint restrictions on the fuzzy operators. In Section \ref{sec:experiments} the experimental results are exposed. The paper ends in Section \ref{sec:conclusions} with some conclusions and future work.


\section{Preliminaries}\label{sec:preliminaries}

In this section, we provide some basic definitions and results related to fuzzy operators that are used throughout this paper. However, we assume that the reader is familiar with basic concepts of fuzzy sets, fuzzy linguistic variables, fuzzy partitions and fuzzy operators (for further information about these topics, the reader can consult \cite{Klir1995,Klement2000,Baczynski2008}).

Fuzzy conjunctions and disjunctions are defined as increasing binary functions $C,D:[0,1]^2 \to [0,1]$ such that $C(0,1)=C(1,0)=0$ and $D(0,1)=D(1,0)=1$, respectively. However, it is common to consider commutative, associative operators with neutral element which are called t-norms and t-conorms.

 \begin{definition}[\bf{\cite{Klement2000}}]
     A \emph{t-norm} (resp. \emph{t-conorm}) is a binary function which is commutative, associative, increasing in both variables and 1 is its neutral element (resp. 0 is its neutral element).
 \end{definition}

\begin{example}
    The following binary functions are t-norms:
    \begin{itemize}
        \item \emph{minimum:} $\TM(x,y)=\min\{x,y\}$.
        \item \emph{algebraic product:} $\TP(x,y)=x \cdot y$.
        \item \emph{Łukasiewicz:} $\TLK(x,y) = \max \{x+y-1,0\}$.
    \end{itemize}
\end{example}

Further, fuzzy implication functions are the generalization of the logical conditional to fuzzy logic.

\begin{definition}[\bf{\cite{Baczynski2008}}]\label{defimp}
    A binary operator $I:[0,1]^2 \to [0,1]$ is said to be a \textsl{fuzzy implication function} if it satisfies:
    \begin{description}
        \item[(I1)]  $I(x,z)\geq I(y,z)\ $  when  $\ x\leq y$, for all $z\in[0,1]$.
        \item[(I2)]  $I(x,y)\leq I(x,z)\ $  when  $\ y\leq z$, for all $x\in[0,1]$.
        \item[(I3)]  $I(0,0)=I(1,1)=1\ $ and $\ I(1,0)=0.$
\end{description}
 \end{definition}

 From Definition \ref{defimp} it is straightforward to see that if $I$ is a fuzzy implication function then $I(0, x) = I(x, 1) = 1$ for all $x \in [0, 1]$. However, the sections $I(\cdot,0)$ and $I(1,\cdot)$ are not fixed by the definition. In fact, $I(\cdot,0)$ induces a fuzzy negation $N_I$ called the natural negation of $I$.

\begin{example} The following binary functions are fuzzy implication functions:
\begin{itemize}
    \item \emph{Łukasiewicz:} $\ILK(x,y) = \min \{1,1-x+y\}$. \vspace{1mm}
    \item \emph{Goguen:}
    $
				\IGG(x,y)
				=
				\left\{\begin{array}{ll}
					1 & \text{ if } x \leq y, \\
					\frac{y}{x}& \text{ if } x>y.
				\end{array}
				\right.
				$ \vspace{1mm}
    \item \emph{G\"odel:}
    $
    \IGD(x,y)
    =
    \left\{\begin{array}{ll}
        1 & \text{ if } x \leq y, \\
        y& \text{ if } x>y.
    \end{array}
    \right.
    $ 
\end{itemize}
    
\end{example}

Usually, additional conditions on fuzzy implication functions are considered. Among those in the literature which are relevant for this paper we can find the following ones.

\begin{definition} 
Let $T$ be a t-norm and $I$ a fuzzy implication function. Then it is said that $I$ fulfills the 
\begin{itemize}
    \item \emph{left neutrality principle:} 
    $$I(1,y)=y, \quad \text{for all } y \in [0,1]. \eqno {\bf (NP)}$$
    \item \emph{ordering property:} 
    $$I(x,y)=1 \Leftrightarrow x \leq y, \quad \text{for all } x,y \in [0,1]. \eqno {\bf (OP)}$$
    \item \emph{$T$-conditionality with respect to $T$:} 
    $$T(x,I(x,y)) \leq y, \quad \text{for all } x,y \in [0,1]. \eqno {\bf (TC)}$$
    
\end{itemize}
\end{definition}

Since Definition \ref{defimp} is quite general, it allows the existence of many different families of fuzzy implication functions.  Indeed, in \cite{Fernandez-Peralta2024} more than 100 families of fuzzy implication functions were gathered. Again, we only provide here the definition of those families related to the results of this paper.

\begin{definition}[\textbf{\cite{Klir1995}}]\label{def:rimplications}
	A function $I:[0,1]^2 \to [0,1]$ is called an \emph{$R$-implication} if there exists a t-norm $T$ such that
	\begin{equation*}
	I(x,y)=\sup \{t \in [0,1] \mid T(x,t) \leq y\}, \quad x,y \in [0,1].
	\end{equation*}
	If $I$ is an $R$-implication generated from a t-norm $T$, then it is denoted by $I_T$.
\end{definition}

\begin{proposition}[\textbf{\cite[Definition 2.10]{Baczynski2016}}] Let $C$ be a copula. The function $I_C: [0,1]^2 \to [0,1]$ given by
$$I_C(x,y)
=
\left \{\begin{array}{ll} 
	1& \hbox{if } x=0,\\
	\frac{C(x,y)}{x} &\hbox{if } x>0,
\end{array}\right.
$$
is a fuzzy implication function if and only if
\begin{equation}\label{eq:condition_copula}
C(x_1,y)x_2 \geq C(x_2,y)x_1,
\end{equation}
for all $x_1 \leq x_2 \text{ and } y \in [0,1]$. In this case, $I_C$ is called a \emph{probabilistic implication}.
\end{proposition}

\begin{definition}[\bf{\cite[Definition 8]{Zhou2021}}] 
	Let $k:[0,1]\to[0,1]$ be a strictly increasing and continuous function with $k(1)=1$. The function $I_k : [0,1]^2 \to [0,1]$ defined by
	$$I_k(x,y)=k^{(-1)}\left( \frac{1}{x} \cdot k(y)\right), \quad x,y\in[0,1],$$ 
	with the understanding $\frac{0}{0}=1$ and $\frac{1}{0}=+\infty$ and where $k^{(-1)}$ is the pseudo-inverse of $k$, is called a \emph{$k$-generated implication}. The function $k$ itself is called a \emph{$k$-generator}.
\end{definition}


\section{Fuzzy implicative rules and quality measures}\label{sec:modeling}

In this section, we recall the modeling of fuzzy implicative rules proposed in \cite{Hullermeier2001} and we study which properties should the pair $(T,I)$ satisfy in order to obtain the desired behavior. Also, we provide the generalization of four of the most important quality measures (coverage, support, confidence, and weighted relative accuracy). Finally, we prove that our framework can be seen as the generalization of the crisp setting and fuzzy conjunctive rules proposals.

\subsection{Rule's Modeling}\label{subsection:rules_modeling}

We consider a set of $n_f \in \mathbb{N}$ features $ \{X_m \mid m=1,\dots,n_f\}$, each with domain $D_m \subseteq \mathbb{R}$ and a target variable $Y$ with domain $D_Y \subseteq \mathbb{R}$. These variables can be categorical or numerical (including the target variable). A set of $n_e \in \mathbb{N}$ examples  $E=\{E^d = (e_1^d, \dots, e_{n_f}^d,y^d) \mid d=1,\dots,n_e\}$.
For each feature $X_m \in \{X_1,\dots,X_{n_f}\}$ we consider $l_m \in \mathbb{N}$ linguistic labels
$X_m: \{LL_m^1, \dots, LL_m^{l_m}\}$.
For each feature $X_m \in \{X_1,\dots,X_{n_f}\}$ and linguistic label $LL_{m}^n$, $n \in \{1,\dots,l_m\}$  we consider a membership function $\mu_{LL_{m}^n}:D_m \to [0,1]$. Moreover, for the target variable we consider $n_c \in \mathbb{N}$ linguistic labels that we call ``classes''
$Y:\{\text{Class}_1, \dots, \text{Class}_{n_c}\}$.
For each $\text{Class}_j$, $j \in \{1,\dots,n_c\}$ we consider a membership function $\mu_{Class_j}:D_Y \to [0,1]$.

For simplicity, we consider rules as the conjunction of literals, for other modelings an analogous study could be made. In this case, to determine a rule we fix a consequent $Class_j$ and a subset of features $S=\{X_{m_1},\dots,X_{m_s}\} \subseteq \{X_1,\dots,X_{n_f}\}$ that will be considered in the antecedent and $L=\{LL^{n_{m_1}}_{m_1},\dots, LL^{n_{m_s}}_{m_s}\}$ the corresponding linguistic labels  with $n_{m_i} \in \{1,\dots,l_m\}$ for all $X_m \in S$, i.e., only a linguistic label per considered feature in the rule. Then, a rule can be expressed as
\begin{eqnarray*}\label{eq:GenericLinguisticRule}
R^{S,L}_j : \text{IF } (X_{m_1} \text{ IS } LL_{m_1}^{n_{m_1}} \text{ AND } \dots  \text{ AND } X_{m_s} \text{ IS } LL_{m_s}^{n_{m_s}}) \\ \nonumber
\text{ THEN } Y \text{ IS } \text{Class}_j.
\end{eqnarray*}

Now, to model the corresponding conjunction and conditional we consider a t-norm $T:[0,1]^2 \to [0,1]$ and a fuzzy implication function $I:[0,1]^2 \to [0,1]$, respectively. We will denote by $\mathcal{R}^{I,T}$ the set of all the rules where the conjunction is modeled by a t-norm $T$ and the conditional by a fuzzy implication function $I$, and by $\mathcal{S}(\mathcal{R}^{I,T})$ the set of all possible subsets of $\mathcal{R}^{I,T}$. Besides, we will denote by $\mathcal{R}^{I,T}_j$ the subset of $\mathcal{R}^{I,T}$ obtained by fixing Class$_{j}$ as the consequent and by $\mathcal{S}(\mathcal{R}_j^{I,T})$ the set of all possible subsets of $\mathcal{R}^{I,T}_j$.

Now, given a rule $R^{S,L}_j \in \mathcal{R}^{I,T}$ expressed as in Eq. (\ref{eq:GenericLinguisticRule}) and an example $E^d \in E$ we define:
\begin{itemize}
	\item The \emph{truth value of the antecedent} of the rule $R^{S,L}_j \in \mathcal{R}^{I,T}$ evaluated on example $E^d$ as
	$$\mu_{ant}^{R^{S,L}_j,E^d}=T\left(\mu_{LL_{m_1}^{n_{m_1}}}(e_{m_1}^d), \dots, \mu_{LL_{m_s}^{n_{m_s}}}(e_{m_s}^d) \right).$$
	\item The \emph{truth value of the consequent} of the rule $R^{S,L}_j \in \mathcal{R}^{I,T}$ evaluated on example $E^d$ as
	$$\mu_{con}^{R^{S,L}_j,E^d}=\mu_{Class_j}(y^d).$$
	\item The \emph{truth value of the rule} $R^{S,L}_j \in \mathcal{R}^{I,T}$ evaluated on example $E^d$ as
	$$\mu_{rule}^{R^{S,L}_j,E^d}=I(\mu_{ant}^{R^{S,L}_j,E^d},\mu_{con}^{R^{S,L}_j,E^d}).$$
	\item The \emph{truth value of the evaluation} of the rule $R^{S,L}_j \in \mathcal{R}^{I,T}$ evaluated on example $E^d$ as
	\begin{eqnarray*}
	\mu_{eval}^{R^{S,L}_j,E^d} &=&T(\mu_{ant}^{R^{S,L}_j,E^d},\mu_{rule}^{R^{S,L}_j,E^d}) \\
    &=& T(\mu_{ant}^{R^{S,L}_j,E^d},I(\mu_{ant}^{R^{S,L}_j,E^d},\mu_{con}^{R^{S,L}_j,E^d})).
	\end{eqnarray*}
\end{itemize}

The properties of the fuzzy operators are of the utmost importance for the correct behavior of the rule mining technique \cite{Burda2015}. If we do not pay attention to this issue we may encounter inconsistencies in the corresponding model. In accordance, we next discuss some desirable properties of fuzzy implicative rules and we compare them with the crisp case. First of all, let us introduce the concept of a refinement of a fuzzy rule as a new rule, whose antecedent has more restrictions.

\begin{definition}\label{def:SDrefinements}
	Let $I$ be a fuzzy implication function, $T$ a t-norm and $R^{\tilde{S},\tilde{L}}_{\tilde{j}}$, $R^{S,L}_j \in \mathcal{R}^{I,T}$. We say that $R^{\tilde{S},\tilde{L}}_{\tilde{j}}$ is a \emph{refinement} of $R^{S,L}_j$ if and only if $\tilde{j}=j$, $ S \subsetneq \tilde{S} \text{ and } LL_{m}^{n_{m}} \in \tilde{L}$ for all $X_{m} \in S$. In this case, we denote it by $R^{S,L}_j \prec R^{\tilde{S},\tilde{L}}_{\tilde{j}}$.
\end{definition}

In the crisp setting, the number of examples that fulfill the conditions in the antecedent is less in any refinement than in the original rule. In the fuzzy case this property is ensured by the monotonicity of the t-norm, and it is interpreted as follows: if $R^{\tilde{S},\tilde{L}}_{\tilde{j}}$ is a refinement of $R^{S,L}_j$, then for any example $E^d$ the truth value when evaluating it in the antecedent of $R^{\tilde{S},\tilde{L}}_{\tilde{j}}$ is smaller than in the antecedent of $R^{S,L}_j$.

\begin{proposition}\label{prop:SD:MonotonicityANT} 
Let $I$ be a fuzzy implication function, $T$ a t-norm and $Class_j \in \{Class_1,\dots,Class_{n_c}\}$, then
	$$ \mu_{ant}^{R^{\tilde{S},\tilde{L}}_{j},E^d} \leq \mu_{ant}^{R^{S,L}_j,E^d},
    $$
    for all $R^{S,L}_j \prec R^{\tilde{S},\tilde{L}}_{j},~ R^{S,L}_j,R^{\tilde{S},\tilde{L}}_{j}  \in \mathcal{R}^{I,T}_j \text{ and } E^d \in E.
	$ Thus, $\displaystyle \sum_{d=1}^{n_e}\mu_{ant}^{R^{\tilde{S},\tilde{L}}_{j},E^d} \leq \sum_{d=1}^{n_e} \mu_{ant}^{R^{S,L}_j,E^d}$.
\end{proposition}

\begin{proof} Without loss of generality let us consider $S= \{X_{m_1},\dots,X_{m_s}\}$, $L = \{LL_{m_{1}}^{n_{m_{1}}}, \dots LL_{m_{s}}^{n_{m_{s}}} \}$   $\tilde{S}\setminus S = \{X_{m_{s+1}},\dots,X_{m_r}\}$ and $\tilde{L} \setminus L = \{LL_{m_{s+1}}^{n_{m_{s+1}}}, \dots LL_{m_{r}}^{n_{m_{r}}} \}$. By the monotonicity of the t-norm $T$ we have
\begin{eqnarray*}
	\mu_{ant}^{R^{\tilde{S},\tilde{L}}_{j},E^d} &=& T\left(\mu_{LL_{m_1}^{n_{m_1}}}(e_{m_1}^d), \dots, \mu_{LL_{m_s}^{n_{m_s}}}(e_{m_s}^d), \right.\\
 & & \left. \mu_{LL_{m_{s+1}}^{n_{m_{s+1}}}}(e_{m_{s+1}}^d), \dots, \mu_{LL_{m_r}^{n_{m_r}}}(e_{m_r}^d) \right) \\
		&\leq&   T\left(\mu_{LL_{m_1}^{n_{m_1}}}(e_{m_1}^d), \dots, \mu_{LL_{m_s}^{n_{m_s}}}(e_{m_s}^d)  \right) \\
    &=& \mu_{ant}^{R^{S,L}_j,E^d}. \hspace{4cm}\qedhere
\end{eqnarray*}
\end{proof}
On the other hand, in a crisp rule the number of examples that fulfill the conditions in the antecedent of a rule and also belong to the target class is smaller than number of the examples that belong to the target class. In the case of fuzzy rules the analogous fact is not straightforward. Since we have generalized the concept of fulfilling the antecedent and belonging to the target class in terms of the generalized modus ponens, we have to impose the $T$-conditionality to ensure that the truth value when evaluating an example in a rule is smaller than the membership degree of the example to the consequent.

\begin{proposition}
	Let $I$ be a fuzzy implication function and $T$ a t-norm. If $I$ satisfies \TC with respect to $T$ then
	$$ \mu_{eval}^{R^{S,L}_{j},E^d} \leq \mu_{con}^{R^{S,L}_j,E^d}, \quad \text{for all }  R^{S,L}_j \in \mathcal{R}^{I,T} \text{ and } E^d \in E.
	$$
	Thus, $\displaystyle \sum_{d=1}^{n_e}\mu_{eval}^{R^{S,L}_j,E^d} \leq \sum_{d=1}^{n_e}\mu_{con}^{R^{S,L}_j,E^d}$. 
\end{proposition}

\begin{proof} By \TC we have
	$$\mu_{eval}^{R^{S,L}_{j},E^d} = T(\mu_{ant}^{R^{S,L}_j,E^d},I(\mu_{ant}^{R^{S,L}_j,E^d},\mu_{con}^{R^{S,L}_j,E^d})) \leq \mu_{con}^{R^{S,L}_j,E^d}.$$
\end{proof}
Besides, in crisp rules the number of examples that fulfill the conditions in the antecedent of a rule and also belong to the target class is smaller in any refinement than in the original rule. This is because in any refinement there are more conditions in the antecedent, so it is more complex.  Again, in the case of fuzzy rules we do not generally have the analogous property. However, differently from the previous case, in order to obtain the desired property for fuzzy rules we have to impose an additional property of fuzzy implication functions for which we could not find any study about it in the consulted bibliography. This property captures the intuitive idea that the truth value of the inference obtained by applying the generalized modus ponens should be decreasing with respect to the truth value of the antecedent.

\begin{definition}\label{def:(MTC)}
	Let $I$ be a fuzzy implication function and $T$ a t-norm. We say that $I$ satisfies the \emph{monotonicity of the generalized modus ponens} with respect to $T$ if and only if
	$$T(\tilde{x},I(\tilde{x},y)) \leq T(x,I(x,y)), \quad x,\tilde{x},y \in [0,1] \text{ with } \tilde{x} \leq x.  \eqno {\text{\MTC}} $$
\end{definition}

Now, if we impose \MTC we can ensure that the truth value of the evaluation of an example is smaller in any refinement than in the original rule.

\begin{proposition}\label{prop:SD:MonotonicityEVAL}
	Let $I$ be a fuzzy implication function, $T$ a t-norm and $Class_j \in \{Class_1,\dots,Class_{n_c}\}$. If $I$ satisfies \MTC with respect to $T$ then
	$$ \mu_{eval}^{R^{\tilde{S},\tilde{L}}_{j},E^d} \leq \mu_{eval}^{R^{S,L}_j,E^d},
    $$
    $\text{for all }  R^{S,L}_j \prec R^{\tilde{S},\tilde{L}}_{j},~ R^{S,L}_j,R^{\tilde{S},\tilde{L}}_{j}  \in \mathcal{R}^{I,T}_j \text{ and } E^d \in E.$ 
	Thus, $\displaystyle \sum_{d=1}^{n_e}\mu_{eval}^{R^{\tilde{S},\tilde{L}}_{j},E^d} \leq \sum_{d=1}^{n_e} \mu_{eval}^{R^{S,L}_j,E^d}$. 
\end{proposition}

\begin{proof}
	Let us consider  $R^{S,L}_j \prec R^{\tilde{S},\tilde{L}}_{j}$, by Proposition \ref{prop:SD:MonotonicityANT} we know that $\mu_{ant}^{R^{\tilde{S},\tilde{L}}_{j},E^d} \leq \mu_{ant}^{R^{S,L}_j,E^d}$, and since the two rules consider the same target class we have $\mu_{con}^{R^{\tilde{S},\tilde{L}}_{j}} = \mu_{con}^{R^{S,L}_j,E^d}$. Thus, by \MTC we obtain
	\begin{eqnarray*}
		\mu_{eval}^{R^{\tilde{S},\tilde{L}}_{j},E^d} &=&  T(\mu_{ant}^{R^{\tilde{S},\tilde{L}}_{j},E^d},I(\mu_{ant}^{R^{\tilde{S},\tilde{L}}_{j},E^d},\mu_{con}^{R^{S,L}_j,E^d})) \\
		& \leq & T(\mu_{ant}^{R^{S,L}_j,E^d},I(\mu_{ant}^{R^{S,L}_j,E^d},\mu_{con}^{R^{S,L}_j,E^d})) \\
		& = & \mu_{eval}^{R^{S,L}_j,E^d}.
	\end{eqnarray*}
\end{proof}

Notice that the property \MTC was not considered in \cite{Hullermeier2001}, so in that approach it was not guaranteed that the evaluation of a rule was monotone with respect to our definition of refinement of a rule (see Definition \ref{def:SDrefinements}).

Finally, in view of the discussion above, we may conclude that the pair of operators $(T,I)$ should satisfy \TC and \MTC in order to behave adequately when used for fuzzy implicative rule mining.

\begin{definition}\label{def:adequate_pair}
	Let $I$ be a fuzzy implication function and $T$ a t-norm. We say that the pair $(T,I)$ is \emph{adequate for fuzzy implicative rule mining} if $I$ fulfills \TC and \MTC with respect to $T$. 
\end{definition}

In the subsequent sections, we prove that these two properties are not only necessary for the reasons pointed out above, but are also crucial for ensuring other desired properties.

\subsection{Fuzzy implicative rules as a generalization of others}\label{subsection:generalization}

Let us point out that the interpretation of the logical conditional relies completely on the selection of the corresponding fuzzy implication function. Thus, our framework is very flexible. Indeed, in this section we prove that the rule's modeling proposed in Section \ref{subsection:rules_modeling} can be seen as the generalization of the crisp setting and other fuzzy logic perspectives based on conjunctive fuzzy rules. 

In the crisp setting the target variable is categorical and, in this case, we can construct a bijection between the domain of the target variable $D_Y$ with $|D_Y|=n_c$ and the set $\{1,\dots,n_c\}$, so let us assume $D_Y=\{1,\dots,n_c\}$. In this situation, it is clear that we are forced to construct a fuzzy set with a singleton membership function for each possible value of $D_Y$, so for all $j \in D_Y$ we consider $Class_j=j$ and
$$   
\begin{array}{rcl}
	\mu_{Class_j}:D_Y&\longrightarrow&\{0,1\}\\
	\tilde{j}&\longmapsto& \left\{ \begin{array}{ll}
		1 &   \text{if }   \tilde{j}=j, \\
		0 &   \text{if }   \tilde{j}\not = j.
	\end{array}
	\right.
\end{array}
$$
Then, it is obvious that when evaluating a rule for a certain example the only possible values for the consequent are $0$ or $1$, so the only values of the fuzzy implication function that are being used are $I(x,1)$ and $I(x,0)$ for all $x \in [0,1]$. Since $I(x,1)=1$ for all $x \in [0,1]$, only the choice of the natural negation $N_I$ plays a role in this framework. However, if we take into account that in the end we are interested in the evaluation of the rule, if we choose a fuzzy implication function $I$ which satisfies \TC with respect to $T$ we have that
$$T(x,I(x,0)) \leq 0 \Rightarrow T(x,I(x,0)) = 0.$$
Thus, when the consequent is zero the evaluation of the rule is also zero independently from the natural negation of the corresponding fuzzy implication function. In accordance, when a singleton membership function for the target variable is considered, the role of the fuzzy implication function disappears and the truth value of the antecedent of the rule is considered as the truth value of the evaluation of the rule, whenever the consequent is non-zero. Then, by imposing the restriction $\TC$ we generalize other perspectives of fuzzy rules where neither numeric targets nor fuzzy implication functions were considered (see for instance \cite{delJesus2007}).
\begin{proposition}
	Let $I$ be a fuzzy implication function, $T$ a t-norm and \linebreak $Class_j \in \{Class_1,\dots,Class_{n_c}\}$ such that $\mu_{Class_j} : D_Y \to \{0,1\}$. If $I$ satisfies \TC with respect to $T$ then $\mu_{con}^{R^{S,L}_j,E^d} \in \{0,1\}$ and 
	$$ \mu_{rule}^{R^{S,L}_j,E^d} =  			\left\{ \begin{array}{ll}
		N_I(\mu_{ant}^{R^{S,L}_j,E^d}) &   \text{if }   \mu_{con}^{R^{S,L}_j,E^d}=0, \\[1em]
		1 & \text{if } \mu_{con}^{R^{S,L}_j,E^d}=1, 
	\end{array} \right.
	$$
	$$
	\mu_{eval}^{R^{S,L}_j,E^d} = \left\{ \begin{array}{ll}
		0 &   \text{if }   \mu_{con}^{R^{S,L}_j,E^d}=0 , \\[1em]
		\mu_{ant}^{R^{S,L}_j,E^d} & \text{if } \mu_{con}^{R^{S,L}_j,E^d}=1,
	\end{array} \right.$$
for all $R^{S,L}_j \in \mathcal{R}^{I,T}_j \text{ and } E^d \in E$.
\end{proposition}
\begin{proof}
	Straightforward.
\end{proof}
On the other hand, there are other fuzzy logic perspectives that consider fuzzy linguistic variables also for the consequent of the rules but the evaluation of the rules is performed by using a t-norm \cite{Burda2015}, i.e., $\mu_{eval}^{R^{S,L}_j,E^d} = T(\mu_{ant}^{R^{S,L}_j,E^d},\mu_{con}^{R^{S,L}_j,E^d})$. In this case, it is clear that if we select the following fuzzy implication function
\begin{equation}\label{eq:y_implication}
I_Y(x,y) = \left\{ \begin{array}{ll}
1 &   \text{if }   x=0 \text{ or } y = 1, \\[1em]
y & \text{otherwise},
\end{array} \right.
\end{equation}
then our approach is equivalent.

\begin{proposition}\label{prop:equivalence_conjunctive}
	Let $T$ be a t-norm and $I_Y$ the fuzzy implication function in Eq. (\ref{eq:y_implication}). Then, $I_Y$ satisfies \TC and \MTC with respect to $T$ and 
	$$
	\mu_{eval}^{R^{S,L}_j,E^d} = T(\mu_{ant}^{R^{S,L}_j,E^d},\mu_{con}^{R^{S,L}_j,E^d}),$$
for all $R^{S,L}_j \in \mathcal{R}^{I,T}_j \text{ and } E^d \in E$.
\end{proposition}

\begin{proof}
    Let $T$ be a t-norm and $I_Y$ the fuzzy implication function Eq. (\ref{eq:y_implication}), then
    \begin{eqnarray*}
    T(x,I_Y(x,y))
    &=&
    \left\{ \begin{array}{ll}
    x &   \text{if }   x=0 \text{ or } y = 1, \\[1em]
    T(x,y) & \text{otherwise},
    \end{array} \right. \\
    &=&
    T(x,y).
    \end{eqnarray*}
\end{proof}
\subsection{Quality measures}\label{subsec:measures}
One of the key decisions when designing a rule mining algorithm is how to select and order the rules that are more interesting for a concrete goal. Since the goal may not be the same depending on the task or the desires of an expert, a plethora of quality measures have been considered (see \cite{Herrera2011} for an overview). In fact, a quality measure can be generally defined as any function  $q: \mathcal{R}^{I,T} \to \RR$ which assigns a numeric value to each rule $R^{S,L}_j \in \mathcal{R}^{I,T}$ \cite{Atzmueller2015}.

In this section, we give a generalization of four of the most widely studied quality measures in rule mining algorithms: Coverage, Support, Confidence and Unusualness or Weighted Relative Accuracy (WRAcc). In our approach, the overall truth values of the antecedent, the consequent and the evaluation of the rule are used for the computation of the quality measures instead of the corresponding percentages of examples. Although some of these quality measures were already proposed in \cite{Hullermeier2001} and also in other fuzzy perspectives that do not take into account fuzzy implication functions \cite{SDEFSR,Burda2015}, in this section we have adapted their definition to our framework and also we have provided a generalization of the Weighted relative accuracy measure. Further, we justify the correct behavior of these measures with respect to the selected fuzzy operators.

\begin{definition}\label{def:fcoverage} 
	Let $I$ be a fuzzy implication function and $T$ a t-norm, we define the \emph{fuzzy coverage} as the function $FCov : \mathcal{R}^{I,T} \to [0,1]$ given by  
	$$FCov(R^{S,L}_j) = \frac{\displaystyle \sum_{d=1}^{n_e} \mu_{ant}^{R^{S,L}_j,E^d}}{|E|},$$
	and it measures the average truth value of the antecedent for all the examples.
\end{definition}

\begin{definition}\label{def:fsupport} 
	Let $I$ be a fuzzy implication function and $T$ a t-norm, we define the \textit{fuzzy support} as the function $FSupp : \mathcal{R}^{I,T} \to [0,1]$ given by 
	$$
		FSupp(R^{S,L}_j) =\frac{\displaystyle \sum_{d=1}^{n_e} \mu_{eval}^{R^{S,L}_j,E^d}}{|E|},
	$$
	and it measures the average truth value of the evaluation of the rule for all the examples.
\end{definition}

\begin{definition}\label{def:fconfidence}
	Let $I$ be a fuzzy implication function and $T$ a t-norm, we define the \emph{fuzzy confidence} as the function $FConf : \mathcal{R}^{I,T} \to [0,1]$ given by 
	$$
		FConf(R^{S,L}_j) = \frac{FSupp(R^{S,L}_j)}{FCov(R^{S,L}_j)} =  \frac{\displaystyle \sum_{d=1}^{n_e} \mu_{eval}^{R^{S,L}_j,E^d}}{\displaystyle \sum_{d=1}^{n_e} \mu_{ant}^{R^{S,L}_j,E^d}},
	$$
	and it measures the quotient of the overall truth value of the evaluation of the rule and the overall truth value of the antecedent for all examples.
\end{definition}

\begin{definition}\label{def:fwracc}
	Let $I$ be a fuzzy implication function and $T$ a t-norm, we define the \emph{fuzzy weighted relative accuracy} as the function $FWRAcc : \mathcal{R}^{I,T} \to [-1,1]$ given by 
	\begin{align*}
		& FWRAcc(R^{S,L}_j) = \\ 
		& FCov(R^{S,L}_j) \cdot \left(FConf(R^{S,L}_j) - \frac{\displaystyle \sum_{d=1}^{n_e}\mu_{con}^{R^{S,L}_j,E^d}}{|E|}\right), 
	\end{align*}
	and it measures the balance between the fuzzy coverage of the rule and its accuracy gain which is computed as the difference between the fuzzy confidence and the average truth value of the consequent for all the examples.
\end{definition}

It is clear that the above quality measures have different goals \cite{Herrera2011}: the support and coverage are measures of generality since they quantify rules according to the individual patterns of interest covered; the confidence is a measure of precision; the unusualness is a measure of interest since it is designed to quantify the potential interest for the user.

Now, let us point out that the measure of fuzzy coverage is monotone with respect to the refinements, i.e., any refinement of a rule has a lower fuzzy coverage. This fact is also true for the fuzzy support if we impose that the pair ($T$,$I$) satisfies \MTC. This highlights again the importance of the newly introduced property. If we do not consider \MTC, the fuzzy support will not necessarily be monotone, and then it will not be adequate as a measure of generality. \New{Indeed, it would be inadmissible to allow that rules with fewer conditions in the antecedent have less support than more complex rules derived from them, as it contradicts human common sense. Also, without this property, support pruning would not be possible, which clearly affects any optimization process.}

\begin{proposition}\label{prop:SD:MonotonicityFCOV} 
	Let $I$ be a fuzzy implication function, $T$ a t-norm and $Class_j \in \{Class_1,\dots,Class_{n_c}\}$, then
	$$
	FCov(R^{S,L}_j) \geq FCov(R^{\tilde{S},\tilde{L}}_{j}),
	$$
	for all $R^{S,L}_j,R^{\tilde{S},\tilde{L}}_{j}  \in \mathcal{R}^{I,T}_j$ such that $R^{S,L}_j \prec R^{\tilde{S},\tilde{L}}_{j}$.
\end{proposition}

\begin{proof}
	Straightforward from Proposition \ref{prop:SD:MonotonicityANT}.
\end{proof}

\begin{proposition}\label{prop:MonotonicitySD}
	Let $I$ be a fuzzy implication function, $T$ a t-norm and $Class_j \in \{Class_1,\dots,Class_{n_c}\}$. If $I$ satisfies \MTC with respect to $T$ then
	$$FSupp(R^{S,L}_j) \geq FSupp(R^{\tilde{S},\tilde{L}}_{j}),$$
	 for all $R^{S,L}_j,R^{\tilde{S},\tilde{L}}_{j}  \in \mathcal{R}^{I,T}$ such that $R^{S,L}_j \prec R^{\tilde{S},\tilde{L}}_{j}$.
\end{proposition}

\begin{proof}
	Straightforward from Proposition \ref{prop:SD:MonotonicityEVAL}.
\end{proof}





\section{Solutions for adequate fuzzy operators}\label{sec:operator}

In Section \ref{sec:modeling} we point out that a pair $(T,I)$ of t-norm and fuzzy implication function used for modeling fuzzy implicative rules should satisfy two properties: \TC and \MTC. It is clear that there exist families of fuzzy implication functions that satisfy $\TC$, but since we have not found any information about \MTC in the literature, in order to select an adequate pair $(T,I)$ we have to further study this property.

It is well known that the main families of fuzzy implication functions satisfy \NP \cite{Baczynski2008}. In accordance, we  first prove that if $I$ satisfies \NP then \MTC implies \TC. Thus, for families that satisfy \NP we only need to study the new property \MTC for the available solutions of \TC.

\begin{proposition}
	Let $I$ be a fuzzy implication function and $T$ a t-norm. If $I$ satisfies $I(1,y) \leq y$ for all $y \in [0,1]$ and \MTC with respect to $T$, then $I$ also satisfies \TC with respect to $T$. In particular, \NP and \MTC with respect to $T$ imply \TC with respect to $T$.
\end{proposition}

\begin{proof}
	Let $x \in [0,1]$, then by \MTC
	$$T(x,I(x,y)) \leq T(1,I(1,y))=I(1,y)\leq y.$$
\end{proof}

Next, we prove that if a fuzzy implication function satisfies \OP then the verification of \MTC can be simplified to the study of the monotonicity of a family of unary functions.

\begin{proposition}\label{prop:(MTC)&(OP)}
	Let $I$ be a fuzzy implication function and $T$ a t-norm. If $I$ satisfies \OP then $I$ satisfies \MTC if and only if the function	
	$$   
	\begin{array}{rcl}
		f_{y}:[y,1]&\longrightarrow&[0,1]\\
		x&\longmapsto& T(x,I(x,y))
	\end{array}
	$$
	is increasing for all $y \in [0,1)$.
\end{proposition}

\begin{proof}
	Let $I$ be a fuzzy implication function that satisfies \OP and $T$ a t-norm such that $f_y$ is an increasing function for all $y \in [0,1)$. Let us prove that the pair $(T,I)$ satisfies \MTC by distinguishing between three cases.
	\begin{itemize}
		\item If $y \geq x \geq \tilde{x}$ then by \OP
		\begin{eqnarray*}
          T(\tilde{x},I(\tilde{x},y)) &=& T(\tilde{x},1)=\tilde{x} \leq x = T(x,1) \\
          &=& T(x,I(x,y)).
        \end{eqnarray*}
		\item If $x \geq \tilde{x} \geq y$ then since $f_y$ is increasing 
		$$T(\tilde{x},I(\tilde{x},y)) =f_y(\tilde{x}) \leq f_y(x) \leq T(x,I(x,y)).$$
		\item If $x \geq y \geq \tilde{x}$ then by \OP and since $f_y$ is increasing
		\begin{eqnarray*}
			T(\tilde{x},I(\tilde{x},\tilde{y})) &=& T(\tilde{x},1)=\tilde{x} \leq y =T(y,1)\\
   &=&T(y,I(y,y)) =f_y(y) \leq f_y(x)\\
   &=&T(x,I(x,y)).
		\end{eqnarray*}
	\end{itemize}
	If $I$ satisfies \MTC it is straightforward to verify that the function $f_y$ is increasing for all $y \in [0,1)$.
\end{proof}

Having said this, we first gather some results about four families in the literature that satisfy \TC. The families considered in this section have been: the $R$-implications with a left-continuous t-norms \cite[Section 2.5]{Baczynski2008} as one of the most-well known families of fuzzy implication functions (which were originally considered in \cite{Hullermeier2001} for modelling fuzzy implicative association rules); the probabilistic implications which interpret the probability of an implication as the conditional probability \cite{Baczynski2016}; $k$-implications which are generated
by multiplicative generators of continuous Archimedean t-norms \cite{Zhou2021}; and strict/nilpotent $T$-power invariant implications as a family of fuzzy implication functions which satisfy the invariance with respect to the powers of a certain t-norm \cite{Fernandez-Peralta2021,Fernandez-Peralta2022}. The selection of these four families among the great variety available has been twofold: we needed a family that has constructive solutions of the $T$-conditionality for which is relatively easy to later study \MTC and also that has a nice structure for applications. In future studies, the families considered could be enlarged.

\begin{corollary}\label{cor:(TC)}
	The following statements hold:
	\begin{enumerate}[label=(\roman*)]
		\item Let $T$ be a left-continuous Archimedean t-norm and $I_T$ its $R$-implication. Then $I_T$ satisfies \TC with $T$.
		\item Let $I_C$ be a probabilistic implication, then $I_C$ satisfies \TC with \TP.
		\item Let $k$ be a $k$-generator with $k \leq \text{id}_{[0,1]}$, $I_k$ the $k$-generated implication and $T_k$ the t-norm generated by $k$ as its multiplicative generator. Then, $I_k$ satisfies \TC with each t-norm $T$ that is weaker than $T_k$, i.e., $T \leq T_k$.
		\item Let \IT be a strict $T$-power invariant implication. Then \IT satisfies  \TC with respect to $T$ if and only if $\varphi(w)=0$ for all $w \in [0,1)$, and $f(x)=g(y)=0$ for all $x,y \in (0,1)$. In this case, \IT satisfies \TC with respect to any t-norm $T^*$.
		\item 	Let \IT be a nilpotent $T$-power invariant implication and $t$ an additive generator of $T$. Then \IT satisfies  \TC with respect to $T$ if and only if $\varphi(w) \leq t^{-1} \left(t(0)(1-w)\right)$ for all $w \in [0,1)$, $f(x) \leq t^{-1}(t(0)-t(x))$ for all $x \in (0,1)$ and $g(y)=0$ for all $y \in (0,1)$.
	\end{enumerate}
\end{corollary}
\begin{proof}
	\begin{enumerate}[label=(\roman*)]
		\item \cite[Theorem 7.4.8]{Baczynski2008}.
		\item \cite[Proposition 4.3]{Baczynski2016}.
		\item \cite[Proposition 11]{Zhou2021}.
		\item \cite[Proposition 4.44]{Fernandez-Peralta2023}. 
		\item \cite[Proposition 4.73]{Fernandez-Peralta2023}. 
	\end{enumerate}
\end{proof}

Next, we point out particular cases of fuzzy implication functions that belong to one of these five families and satisfy \MTC. Also, we consider the fuzzy implication function in Eq. (\ref{eq:y_implication}) just to recall again that this operator generalizes the fuzzy conjunctive rule perspective.

\begin{proposition}\label{prop:(MTC)} 
The following statements hold:
	\begin{enumerate}[label=(\roman*)]
		\item The pair $(\IGD, \TM)$ satisfies \MTC. In this case,
        \begin{equation}
        \TM(x,\IGD(x,y)) = \min \{x,y\}.
        \end{equation}
		\item Let $T$ be a left-continuous Archimedean t-norm and $I_T$ its $R$-implication. Then $I_T$ satisfies \MTC with $T$. In this case,
        \begin{equation}
        T(x,I_T(x,y)) = \min \{x,y\}.
        \end{equation}
		\item Let $I_C$ be a probabilistic implication, then $I_C$ satisfies \MTC with \TP. In this case,
        \begin{equation}
        \TP(x,I_C(x,y)) = C(x,y).
        \end{equation}
		\item  Let $k$ be a $k$-generator, $I_k$ the $k$-generated implication and $T_k$ the t-norm generated by $k$ as its multiplicative generator. If $k(\tilde{x}) \cdot x \leq k(x) \cdot \tilde{x}$ for all $0\leq \tilde{x} \leq x \leq 1$, then $I_k$ satisfies \MTC with $T_k$. In this case,
        \begin{equation}\label{eq:mtc_kimpl}
        T_k(x,I_k(x,y)) = k^{-1} \left(k(x) \cdot \min\left\lbrace\frac{k(y)}{x},1\right\rbrace\right).
        \end{equation}
		\item Let $T$ be a continuous Archimedean t-norm and $\IT$ the strict/nilpotent $T$-power invariant implication obtained from it. If $\IT$ satisfies \TC and \MTC with $T$ then $T(x,\IT(x,y))=0$ for all $x \in [0,1]$ and $y \in (0,1)$.
        \item  Let $T$ be any t-norm and $I_Y$ the fuzzy implication function given by Eq. (\ref{eq:y_implication}), then $I_Y$ satisfies \MTC with respect to $T$. In this case,
        $$T(x,I_Y(x,y)) = T(x,y).$$
	\end{enumerate}
\end{proposition}

\begin{proof}
	\begin{enumerate}[label=(\roman*)]
		\item  Let $y \in [0,1)$, then
		\begin{eqnarray*}
        \TM(x,\IGD(x,y))
		&=&
		\left\{ \begin{array}{ll}
			\TM(x,1) &   \text{if }   x \leq y , \\
			\TM(x,y) & \text{if } x>y, 
		\end{array} \right. \\
		&=&
		\left\{ \begin{array}{ll}
			x &   \text{if }   x \leq y , \\
			y & \text{if } x>y, 
		\end{array} \right.
		\end{eqnarray*}
		and since \IGD satisfies \OP, by Proposition \ref{prop:(MTC)&(OP)}  we have that \IGD satisfies \MTC with \TM.
		\item Let $T$ be a continuous Archimedean t-norm and $t$ an additive generator of $T$, then the corresponding $R$-implication is given by
		$$ I_T(x,y) = \left\{ \begin{array}{ll}
			1 &   \text{if }   x \leq y , \\
			t^{-1}(t(y)-t(x)) & \text{if } x>y. 
		\end{array} \right.
		$$
		Let $y \in [0,1)$, then
		\begin{eqnarray*}
		T(x,I_T(x,y))
		&=&
		\left\{ \begin{array}{ll}
			\hspace{-1mm} T(x,1) & \hspace{-1mm}  \text{if }   x \leq y , \\
			\hspace{-1mm} T(x,t^{-1}(t(y)-t(x))) & \hspace{-1mm} \text{if } x>y, 
		\end{array} \right. \\
		&=&
		\left\{ \begin{array}{ll}
			\hspace{-1mm} x &  \hspace{-2mm} \text{if }   x \leq y , \\
			\hspace{-1mm} t^{(-1)}(t(x)+t(y)-t(x)) & \hspace{-2mm}\text{if } x>y, 
		\end{array} \right. \\
		&=&
		\left\{ \begin{array}{ll}
			\hspace{-1mm} x &   \text{if }   x \leq y , \\
			\hspace{-1mm} y & \text{if } x>y. 
		\end{array} \right.
		 \end{eqnarray*}
	 	Since $I_T$ satisfies \OP, by Proposition \ref{prop:(MTC)&(OP)}  we have that $I_T$ satisfies \MTC with $T$.
		\item Let $I_C$ be a probabilistic implication, then for all $\tilde{x},x,y \in [0,1]$ such that $0<\tilde{x} \leq x$ we have
		\begin{eqnarray*}
          \TP(\tilde{x},I_C(\tilde{x},y))&=&\tilde{x} \cdot  \frac{C(\tilde{x},y)}{\tilde{x}} = C(\tilde{x},y) \leq C(x,y) \\
          &=&  \TP(x,I_C(x,y)).
        \end{eqnarray*}
		On the other hand, if $\tilde{x}=0$ then $\TP(\tilde{x},I_C(\tilde{x},y))=\TP(0,I_C(0,y)) = 0 \leq \TP(x,I_C(x,y))$ for all $x \in [0,1]$.
		\item According to the proof of \cite[Proposition 11]{Zhou2021} we have 
		$$T_k(x,I_k(x,y)) = k^{-1} \left(k(x) \cdot \min\left\lbrace\frac{k(y)}{x},1\right\rbrace\right).$$
		Let $\tilde{x},x,y \in [0,1]$, if $\tilde{x}=0$ then it is clear that $T_k(\tilde{x},I_k(\tilde{x},y))=T_k(0,I_k(0,y))=0 \leq T_k(x,I_k(x,y))$, so let us consider $0<\tilde{x}\leq x \leq 1$. We distinguish between different cases:
		\begin{itemize}
			\item If $k(y)>x$ then $k(y) >x \geq \tilde{x}$ and
			$$T_k(\tilde{x},I_k(\tilde{x},y)) = \tilde{x} \leq x = T_k(x,I_k(x,y)).$$
			\item If $k(y)> \tilde{x}$ and $k(y)\leq x$, since $\frac{k(\tilde{x})}{\tilde{x}} \leq \frac{k(x)}{x}$ we have
			\begin{eqnarray*}
			T_k(x,I_k(x,y)) &=& k^{-1} \left(\frac{k(x)k(y)}{x}\right) \\
            &\geq& k^{-1} \left(\frac{k(\tilde{x})k(y)}{\tilde{x}}\right) 
			\geq k^{-1}(k(\tilde{x})) \\&=&\tilde{x} = T_k(\tilde{x},I_k(\tilde{x},y)).
			\end{eqnarray*}
		\item If $k(y) \leq \tilde{x}$, since $\frac{k(\tilde{x})}{\tilde{x}} \leq \frac{k(x)}{x}$ we have
		\begin{eqnarray*}
		T_k(\tilde{x},I_k(\tilde{x},y)) &=& k^{-1} \left(\frac{k(\tilde{x})k(y)}{\tilde{x}}\right) \\
        &\leq& k^{-1} \left(\frac{k(x)k(y)}{x}\right) \\
        &=& T_k(x,I_k(x,y)).
		\end{eqnarray*}
		\end{itemize}
		\item Let $T$ be a continuous Archimedean t-norm and $\IT$ the corresponding strict/nilpotent $T$-power invariant implication. If $\IT$ satisfies  \TC with respect to $T$ then by \cite[Propositions 4.44 and 4.73]{Fernandez-Peralta2023} we obtain that $\varphi(0^+)=0$ and by \cite[Proposition 3]{Fernandez-Peralta2021} and \cite[Proposition 20]{Fernandez-Peralta2022} we have that $\IT(1,y)=g(y)=0$ for all $y \in (0,1)$. Then, since \IT satisfies \MTC with $T$ we have
		\begin{eqnarray*}
        T(x,\IT(x,y)) &\leq & T(1,\IT(1,y))= T(1,g(y))\\
        &=&T(1,0)=0,
        \end{eqnarray*}
		which implies $T(x,\IT(x,y)) = 0$ for all $x \in [0,1]$ and $y \in (0,1)$.
        \item It is proved in Proposition \ref{prop:equivalence_conjunctive}.
	\end{enumerate}
\end{proof}

From Proposition \ref{prop:(MTC)} we can conclude the following:

\begin{itemize}
	\item  If $T$ is the minimum or a continuous Archimedean t-norm then it satisfies \MTC with the corresponding $R$-implication, $I_T$. Moreover, by Corollary \ref{cor:(TC)} we know that in this case $(T,I_T)$ also satisfies \TC. Thus, we deduce that $(T,I_T)$ is an adequate pair for fuzzy implicative rule mining when $T$ is the minimum or a continuous Archimedean t-norm. Notice that we have not studied \MTC when $T$ is different from the minimum or a continuous Archimedean t-norm. These cases could be performed in future studies of this property.
	\item If $I_C$ is a probabilistic implication, then it satisfies \MTC with respect to \TP. Thus, $(\TP,I_C)$ is an adequate pair for fuzzy implicative rule mining.
	\item  If $k$ is a $k$-generator with $k(\tilde{x}) \cdot x \leq k(x) \cdot \tilde{x}$ for all $0 \leq \tilde{x} \leq x \leq 1$ then $(T_k,I_k)$ where $I_k$ is the $k$-generated implication and $T_k$ the t-norm generated by $k$ as its multiplicative generator satisfy \MTC. In this case, if we consider $x=1$ then we have $k(\tilde{x}) \leq k(1) \cdot \tilde{x} = \tilde{x} = \text{id}(\tilde{x})$ for all $0 \leq \tilde{x} \leq 1$ and the pair $(T_k,I_k)$ also satisfies \TC. In particular, $k(0)=0$ in this case. Then, $(T_k,I_k)$ is an adequate pair for fuzzy implicative rule mining under these conditions. For instance, let us consider the following $k$-generators:
	$$k_{\lambda}^{SS}(x) = e^{\frac{x^{\lambda}-1}{\lambda}}, \quad \lambda \in (-\infty,0),$$
	$$k_{\lambda}^{H}(x) = \frac{x}{\lambda + (1-\lambda)x}, \quad \lambda \in (1,+\infty),$$
	$$k_{\lambda}^{F}(x) = \frac{\lambda^x-1}{\lambda-1}, \quad \lambda \in (1,+\infty).$$
	Then, we consider the functions $f_{k_{\lambda}}^L(x)=\frac{k_{\lambda}^{L}(x)}{x}$ for all $x \in (0,1]$ and $L \in \{SS,H,F\}$ and we compute their first derivative:
	$$(f_{k_{\lambda}}^{SS})'(x)= \frac{e^{\frac{x^{\lambda}-1}{\lambda}} \cdot (x^{\lambda}-1)}{x^2}, \quad \lambda \in (-\infty,0),$$
	$$(f_{k_{\lambda}}^{H})'(x)= \frac{\lambda-1}{(\lambda+(1-\lambda)x)^2}, \quad \lambda \in (1,+\infty),$$
	$$(f_{k_{\lambda}}^{F})'(x)= \frac{\lambda^x (x \ln \lambda -1)+1}{x^2 (\lambda-1)}, \quad \lambda \in (1,+\infty).$$
	It is straightforward to verify that $(f_{k_{\lambda}}^{L})'(x)>0$ for all $x \in (0,1]$, so all $f_{k_{\lambda}}^{L}$ are increasing and consequently, $k_{\lambda}^L(\tilde{x}) \cdot x \leq k_{\lambda}^{L}(x) \cdot \tilde{x}$ for all $0 \leq \tilde{x} \leq x \leq 1$,  $L \in \{SS,H,F\}$ and the corresponding domain of $\lambda$. Therefore, all the pairs $(T_{k_{\lambda}^L},I_{k_{\lambda}^L})$ are adequate for fuzzy implicative rule mining (in \cite[Example 1]{Zhou2021} the reader can find the corresponding constructions).
	\item The point (v)-Proposition \ref{prop:(MTC)} discloses that the family of $T$-power invariant implications does not provide interesting solutions when studying the property \MTC. Indeed, any solution of this property returns zero when applying the generalized modus ponens. Thus, for the moment, there are no solutions if we want to use fuzzy implication functions that are invariant with respect to hedges modeled through powers of t-norms.
    \item Finally, if we consider any t-norm and the fuzzy implication function in Eq. (\ref{eq:y_implication}), then we obtain a pair of operators that is adequate for fuzzy implicative rule mining, and the modeling is equivalent to fuzzy conjunctive rules.
\end{itemize}

\New{
Once we have several options for the selection of the pair of operators $(T,I)$ it is interesting to look at the structure of the evaluation of the antecedent and the rule for different pairs. In Proposition \ref{prop:(MTC)} we can see the expression of the truth value of the evaluation of the rule when the generalized modus ponens is used for each case. Notice that even though fuzzy implication functions are non-commutative, it may happen that the expression given by the generalized modus ponens is commutative. Indeed, this is the case for $R$-implications and also probabilistic implications that use a commutative copula. What is more, any choice of $R$-implication, one of the most popular fuzzy implication functions in fuzzy inference, is equivalent to using the minimum t-norm, and then it does not provide useful and different information for our framework. Thus, the user might like to avoid these options if the purpose is to obtain directional rules, in which the permutation of antecedent and consequent is significantly different.  Then, we have seen that it may be hard to find a pair of operators $(T,I)$ so the value $T(x,I(x,y))$ fulfills all the desired properties. Indeed, although probabilistic implications are an option, it is not so easy to find a non-commutative copula fulfilling Eq. (\ref{eq:condition_copula}). Besides, $k$-implications reduce the value of that expression to 1 in a certain subdomain of $[0,1]^2$ (see Eq. (\ref{eq:mtc_kimpl})). Nonetheless, it is not hard to propose a pair $(T,I)$ that fulfills all the conditions and, in contrast with those operators already introduced in the literature, it may be more tailored for a specific problem. For instance, we propose the following fuzzy implication function.
\begin{proposition}\label{prop:Ip} 
Let $p \in (0,+\infty)$ and $I_p :[0,1]^2 \to [0,1]$ defined as follows:
	\begin{equation}\label{eq:Ipq} 
		I_{p}(x,y) = 1-x + xy ^p, \quad \text{for all } x,y \in [0,1]. 
	\end{equation}
Then $I_p$ is a fuzzy implication function with $\TLK(x,I_p(x,y)) = xy^p$ and the pair $(\TLK,I_p)$ fulfills \TC and \MTC.
\end{proposition}
\begin{proof}
	First, we prove that $I_p$ is a fuzzy implication function. Indeed, $I_p(0,0)=I_p(1,1)=0$, $I_p(1,1)=1$ and
	$$
	x_1 \leq x_2 \Rightarrow 1+x_1(y^p-1) \geq 1+x_2(y^p-1),
	$$
	$$
	y_1 \leq y_2 \Rightarrow 1+x(y_1^p-1) \leq 1+x(y_2^p-1),
	$$
	so we obtain $_pI(x_1,y) \geq I_p(x_2,y)$ and $I_p(x,y_1) \leq I_p(x,y_2)$.
	$$
	\TLK(x,I_p(x,y)) = \max \{x+1-x+xy^p-1,0\} = xy^p.
	$$
	Now, since $xy^p \leq x$ for all $x,y \in [0,1]$ and $x_1y^p \leq x_2y^p$ we know that the pair $(\TLK,I_p)$ fulfills \TC and \MTC.
\end{proof}
With the pair of operators in Proposition \ref{prop:Ip}, we can adjust the importance of the consequent by modifying the parameter $p$; if $p>1$, we shrink its contribution, whereas when $p<1$, we inflate it. Moreover, the functional form of the operators is simple and easily implementable. We see in Section \ref{subsec:directional_patterns} how this non-commutative behavior is a key point when trying to detect directional patterns.}

\New{
Summarizing the above discussion, in the following proposition, we point out six different pairs of fuzzy implication functions and t-norms that are adequate for fuzzy implicative rule mining. More specifically, in Table \ref{table:selectedoperators} we propose four concrete examples of adequate pairs, which are the ones considered in our experimental results.
\begin{proposition}\label{prop:adequate_pairs}
The following pairs are adequate for fuzzy implicative rule mining in the sense of Definition \ref{def:adequate_pair}:
	\begin{itemize}
		\item $(\TM,\IGD)$.
		\item $(T,I_{\bf T})$ where $T$ is a continuous, Archimedean t-norm.
		\item $(\TP,I_C)$ where $I_C$ is a probabilistic fuzzy implication.
		\item $(T_k,I_k)$ where $k$ is a $k$-generator with $k(\tilde{x})\cdot x \leq k(x) \cdot \tilde{x}$ for all $0\leq \tilde{x} \leq x \leq 1$, $I_k$ is a $k$-generated implication and $T_k$ the t-norm generated by $k$ as its multiplicative generator.
        \item $(T,I_Y)$ where $I_Y$ is the fuzzy implication function in Eq. (\ref{eq:y_implication}) and $T$ is any t-norm.
        \item $(\TLK,I_p))$, where $I_p$ is the fuzzy implication function in Eq. (\ref{eq:Ipq}).
	\end{itemize}
\end{proposition}
\begin{table}[!htbp]
	\caption{Four examples of pairs $(T,I)$ which fulfill \TC and \MTC.}\label{table:selectedoperators}
	\centering
	\setlength\tabcolsep{5pt}
	\renewcommand{\arraystretch}{1.3} \large
	\resizebox{9cm}{!}{
		\begin{tabular}{|l|l|}
	\hline
	\bf t-norm $\bm{T}$ & \bf Fuzzy implication function $\bm{I}$\\ \hline
	\begin{tabular}[c]{@{}l@{}} \textit{Product t-norm} \\ $\TP(x,y)=xy$.   \end{tabular} & \begin{tabular}[c]{@{}l@{}}
		$
		I_Y(x,y)
		=
		\left\{\begin{array}{ll}
			1 & \text{if } x=0 \text{ and } y=1, \\
			y& \text{otherwise}.
		\end{array}
		\right.
		$\end{tabular}       \\ \hline
	\begin{tabular}[c]{@{}l@{}} \textit{Łukasiewicz t-norm} \\ $\TLK(x,y)=\max\{x+y-1,0\}$.  \end{tabular} & 	\begin{tabular}[c]{@{}l@{}}			\textit{Łukasiewicz fuzzy implication} \\
		$
		\ILK(x,y)
		= \min\{1,1-x+y\}.
		$ \end{tabular}      \\ \hline
			\begin{tabular}[c]{@{}l@{}l@{}} \textit{Schweizer-Sklar t-norm} \\ $
	T_{\lambda}^{\bm{SS}}(x,y)
	=
	(\max\{(x^{\lambda}+y^{\lambda}-1),0\})^{\frac{1}{\lambda}}.
	$  \end{tabular}  & \begin{tabular}[c]{@{}l@{}l@{}}
	\textit{$k_{\lambda}$-generated Schweizer-Sklar implications} \\
	$
	I^{\bm{SS}}_{k_{\lambda}}(x,y)
	=
	\left\{\begin{array}{ll}
		1 & \text{ if } x \leq e^{\frac{y^{\lambda}-1}{\lambda}}, \\
		(y^{\lambda}-\lambda \ln x)^{\frac{1}{\lambda}} & \text{otherwise,}
	\end{array}
	\right.
	$\\
	with $\lambda \in (-\infty,0)$.
\end{tabular}   \\ \hline
	\begin{tabular}[c]{@{}l@{}} \textit{Łukasiewicz t-norm} \\ $\TLK(x,y)=\max\{x+y-1,0\}$.  \end{tabular} & 	\begin{tabular}[c]{@{}l@{}}
	$
	I_p(x,y) = 1-x+xy^p.
	$ \end{tabular}     \\ \hline
\end{tabular}
}
\end{table}
}

\section{Experiments}\label{sec:experiments}
\New{In this section, we discuss the role of fuzzy implicative rules in the specific problem of mining association rules, which we use as an illustrative case to demonstrate the utility of the framework. For simplicity, we rely on the well-known exhaustive Apriori algorithm \cite{Agrawal1993}, enhanced with redundancy pruning as described in \cite{Hullermeier2003}. Apriori identifies frequent itemsets based on a minimum fuzzy coverage threshold and subsequently derives association rules by applying minimum thresholds for fuzzy support and confidence. Unlike classical approaches, in our framework support is not necessarily symmetric, and this asymmetry is explicitly taken into account when formulating the final rules. Consequently, for each frequent itemset, we must consider all possible directional antecedent–consequent splits. Accordingly, the rule-generation cost increases with the frequent itemset size, yet empirically, it did not dominate the runtime for the datasets evaluated. The frequent itemset generation cost increases exponentially with the number of items, just as classical Apriori. Nonetheless, coverage and support thresholds significantly lower the computation time. The Python implementation is available in the open repository \url{https://github.com/rferper/FIRM}
, where we allow the user to select custom fuzzy partitions and operators. Moreover, as we proved in Section~\ref{subsection:generalization}, our framework generalizes other fuzzy approaches, so we use our implementation to compare with existing fuzzy association mining tools such as the \texttt{lfl} package \cite{Burda2022}. For comparison with crisp association rules, we use the \texttt{arules} package \cite{hahsler2023}. In this case, discretization is performed using the same cuts as the fuzzy sets to ensure a fair evaluation.} 
\New{\subsection{Mining directional patterns}\label{subsec:directional_patterns}
This section presents a controlled experiment illustrating how fuzzy association rules may or may not capture directional dependencies. We first generated a synthetic bivariate dataset designed to encode an asymmetric relation. Variable \(A\) was sampled uniformly on \([0,100]\). Conditional on \(A\), variable \(B\) followed
\[
B \sim 
\begin{cases}
	\mathcal{U}(0,100), & A \le 70,\\[4pt]
	\mathcal{U}(70,100), & A > 70,
\end{cases}
\]
so that large values of \(A\) enforce correspondingly large values of \(B\). 
This mechanism creates a structured gap in the joint distribution. In total, we generated 1000 points and the resulting scatter plot is shown in Figure~\ref{figure:synthetic_data}.
\begin{figure}[ht]
	\centering
	\includegraphics[width=0.7\columnwidth]{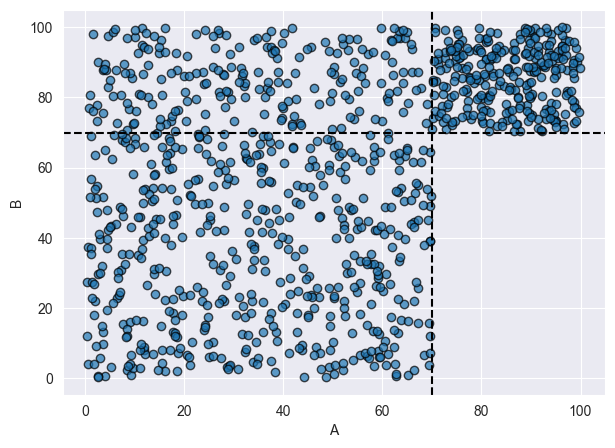}
	\caption{Synthetic dataset illustrating an induced asymmetric dependency between two variables $A$ and $B$.}
	\label{figure:synthetic_data}
\end{figure}
From this construction, it follows that if $A$ is High then $B$ is High, but not necessarily the other way around. Moreover, no other rules are expected to hold significance.}

\New{
We now examine how support and confidence evolve for operator pairs exhibiting non-commutative behavior. Figure~\ref{fig:fourgrid} shows that, for $(T_{\lambda}^{\bm{SS}},I^{\bm{SS}}_{k_{\lambda}})$, support and confidence remain nearly invariant, with both rules showing very similar support. Although $A$ High $\to$ $B$ High achieves slightly higher confidence than the reverse, the difference is negligible. Hence, these operators, despite being non-commutative, are too restrictive to capture the intended directional pattern. In contrast, for $(\TLK,I_p)$ the behavior changes substantially. For $p<1$, support and confidence diverge clearly between the two rules, correctly highlighting the true dependency. For $p>1$, however, the opposite occurs: the rule $B$ High $\to$ $A$ High gains inflated confidence due to overemphasizing the antecedent, which is frequently satisfied when $B$ is High.
\begin{figure}[ht]
	\centering
	\begin{subfigure}[b]{\linewidth}
		\centering
		\includegraphics[width=0.48\linewidth]{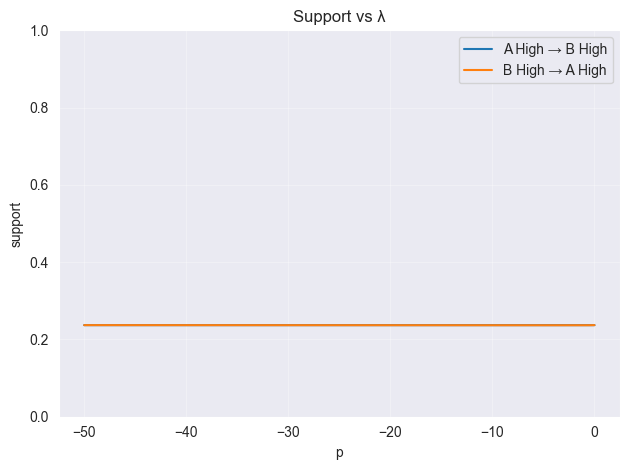}%
		\hfill
		\includegraphics[width=0.48\linewidth]{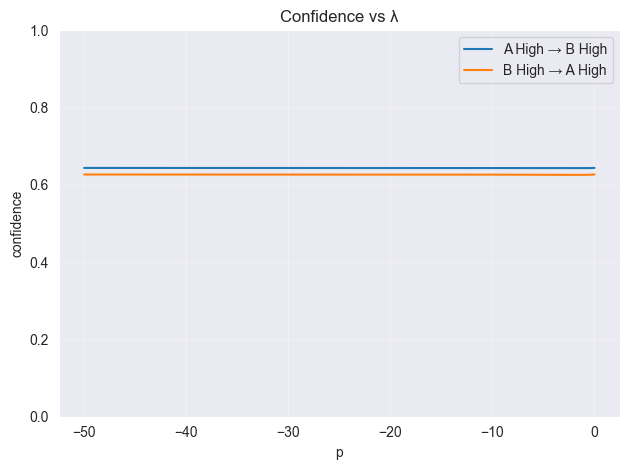}
		\subcaption{$(T_{\lambda}^{\bm{SS}},I^{\bm{SS}}_{k_{\lambda}})$.}
		\label{fig:rowsub:1}
	\end{subfigure}
	\medskip
	\begin{subfigure}[b]{\linewidth}
		\centering
		\includegraphics[width=0.48\linewidth]{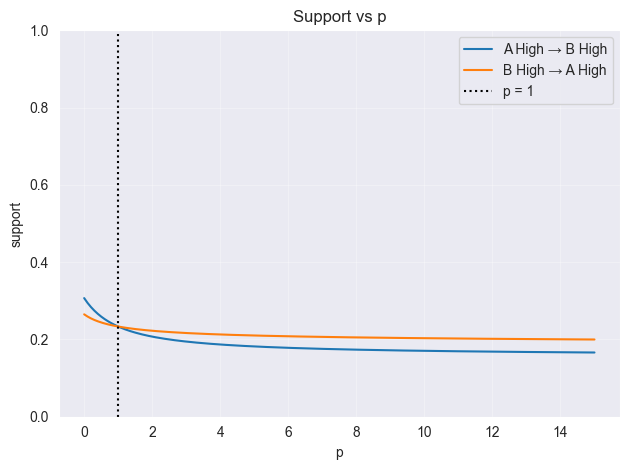}%
		\hfill
		\includegraphics[width=0.48\linewidth]{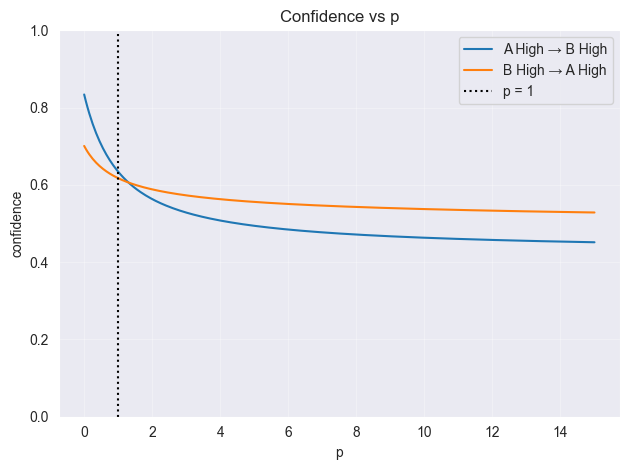}
		\subcaption{$(\TLK,I_p)$.}
		\label{fig:rowsub:2}
	\end{subfigure}
	\caption{Evolution of support and confidence for the rules $A$ High $\to$ $B$ High and $B$ High $\to$ $A$ High as a function of the parameters $\lambda$ and $p$ of the pairs $(T_{\lambda}^{\bm{SS}},I^{\bm{SS}}_{k_{\lambda}})$ and $(\TLK,I_p)$, respectively.}
	\label{fig:fourgrid}
\end{figure}
}

\New{For comparison, we also computed support and confidence for fuzzy association rules using both the minimum and product t-norms, as well as for the crisp case. The results are reported in Table~\ref{table:commutative_case}. Across all three settings, the values are nearly identical and do not distinguish between the two rules. As expected, support remains identical for both rules in all cases.
\begin{table}[ht]
	\centering
	\caption{Support and confidence values for the rules $A$ High $\to$ $B$ High and $B$ High $\to$ $A$ High under different operator choices: product t-norm with Yager implication $(T_P,I_Y)$, Łukasiewicz t-norm with Łukasiewicz implication $(T_{LK},I_{LK})$, and the crisp case.}\label{table:commutative_case}
	\begin{tabular}{c|c|c|}
		\cline{2-3}
		& $A$ High $\to$ $B$ High & $B$ High $\to$ $A$ High \\ \hline
		\multicolumn{1}{|c|}{$\TP$} & (0.23,0.63) &  (0.23,0.62) \\ \hline
		\multicolumn{1}{|c|}{$\TM$} & (0.24,0.64) & (0.24,0.63) \\ \hline
		\multicolumn{1}{|c|}{Crisp} & (0.23,0.64) & (0.23,0.62) \\ \hline
	\end{tabular}
\end{table}}
\New{In summary, we have shown that for recovering the true asymmetric dependency, we require operators with non-symmetric support and explicit control over the consequents’ contribution. Both properties become available through the introduction of fuzzy implication functions.}

\subsection{Similarity between models}

\New{To assess whether fuzzy implicative rules yield knowledge that differs from other approaches, we conducted experiments on seven benchmark datasets from the UCI repository \cite{Dua2019}: \textit{iris}, \textit{wdbc}, \textit{vehicle}, \textit{abalone}, \textit{magic}, \textit{onlinenews} and \textit{globalhousing}. These datasets include both numerical and categorical variables. Unless specified otherwise, for the numerical variables, we constructed fuzzy partitions consisting of three triangular fuzzy sets defined according to the variable quantiles. For the categorical variables, we used crisp fuzzy partitions aligned with the number of categories. For every database, we have selected 0.3 as the minimum fuzzy coverage and support and 0.8 as the minimum fuzzy confidence, as we found it is a good balance between quality and number of rules (see Figure \ref{fig:sensitive_analysis}). As baselines for non-implicative fuzzy rule mining, we use the minimum \TM and the product \TP t-norms; for our proposed $(T,I)$ semantics we consider $(T_{\lambda}^{\bm{SS}}, I^{\bm{SS}}_{k_{\lambda}})$ with $\lambda=-10$ and $(\TLK, I_p)$ with $p=0.01$. The results can be found in Table \ref{table:results1}.}

\begin{figure}[t]
  \centering
  \begin{subfigure}{0.48\linewidth}
    \includegraphics[width=\linewidth]{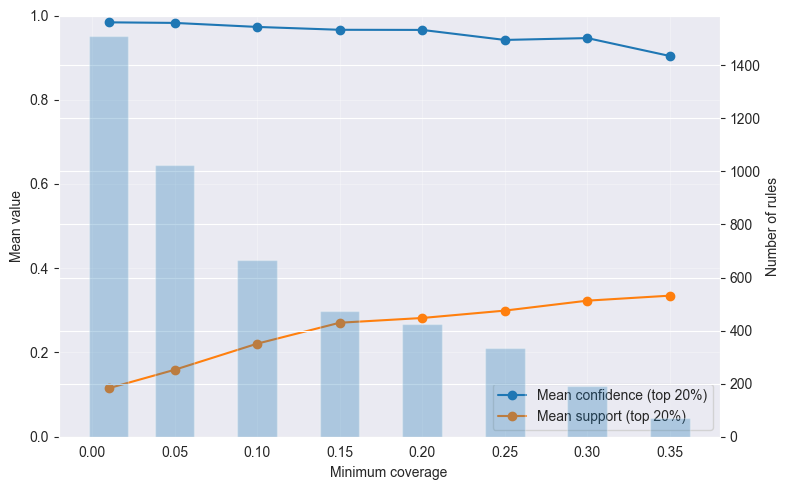}
    \caption{iris}
    \label{fig:rules_cov_dsA}
  \end{subfigure}\hfill
  \begin{subfigure}{0.48\linewidth}
    \includegraphics[width=\linewidth]{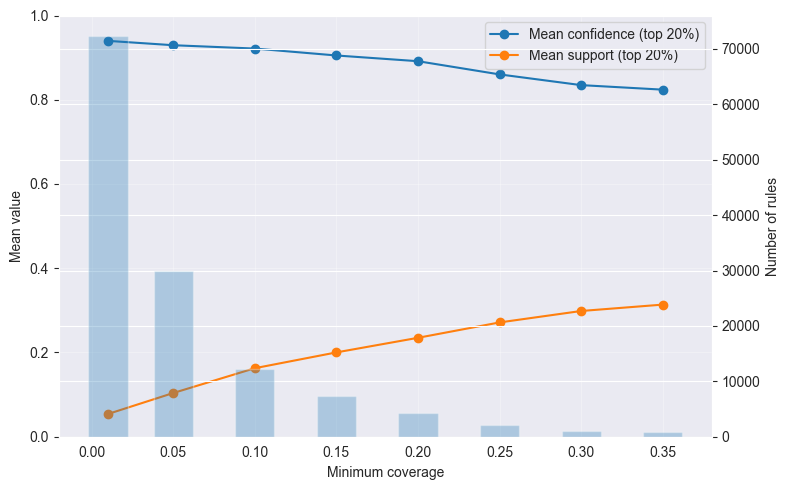}
    \caption{wdbc}
    \label{fig:rules_cov_dsB}
  \end{subfigure}
  \caption{Effect of minimum fuzzy coverage on rule quality and set size (top 20\% of rules by confidence). Left axis: mean confidence and mean support; right axis: number of rules. Fuzzy operators were fixed to $(\ILK,I_p)$ with $p=0.01$.}
  \label{fig:sensitive_analysis}
\end{figure}

\begin{table}[!htbp]
\caption{Results per dataset and model. In each cell, it appears the number of rules (top) and a tuple with the mean of the four quality measures Coverage, Support and Confidence. Also, under the name of each dataset there is a summary of its specifications: (samples, features).} \label{table:results1}
	\centering
	\setlength\tabcolsep{7pt}
	\renewcommand{\arraystretch}{2.1} \large
	\resizebox{9cm}{!}{
		\begin{tabular}{|c|c|c|c|c|c|}
			\hline
			\diagbox[width=10em]{\textbf{Dataset}}{\textbf{Model}}   
                & $\TP$ 
                & $\TM$ 
                & $(T_{-10}^{\bm{SS}},I^{\bm{SS}}_{k_{-10}})$ 
                & $(\TLK,I_{0.01})$
                & \textbf{Crisp} \\ \hline
\multicolumn{1}{|c|}{\begin{tabular}[c]{@{}c@{}}\textbf{iris}\\ (150,4) \end{tabular}}    
    & \multicolumn{1}{c|}{\begin{tabular}[c]{@{}c@{}}17\\ (0.35,0.32,0.92)\end{tabular}}   
    & \multicolumn{1}{c|}{\begin{tabular}[c]{@{}c@{}}19\\ (0.35,0.32,0.92)\end{tabular}}   
    & \multicolumn{1}{c|}{\begin{tabular}[c]{@{}c@{}}15\\ (0.35,0.32,0.92) \end{tabular}}  
    & \multicolumn{1}{c|}{\begin{tabular}[c]{@{}c@{}}32\\ (0.35,0.33,0.96) \end{tabular}}  
    & \multicolumn{1}{c|}{\begin{tabular}[c]{@{}c@{}}39\\ (0.33,0.31,0.95) \end{tabular}}  \\ \hline
\multicolumn{1}{|c|}{\begin{tabular}[c]{@{}c@{}}\textbf{wdbc}\\ (569,11) \end{tabular}}    
    & \multicolumn{1}{c|}{\begin{tabular}[c]{@{}c@{}}49\\ (0.36,0.32,0.9)\end{tabular}} 
    & \multicolumn{1}{c|}{\begin{tabular}[c]{@{}c@{}}51\\ (0.36,0.34,0.93) \end{tabular}} 
    & \multicolumn{1}{c|}{\begin{tabular}[c]{@{}c@{}}39\\ (0.37,0.32,0.88) \end{tabular}} 
    & \multicolumn{1}{c|}{\begin{tabular}[c]{@{}c@{}}84\\ (0.36,0.34,0.93) \end{tabular}} 
    & \multicolumn{1}{c|}{\begin{tabular}[c]{@{}c@{}}67\\ (0.37,0.34,0.93) \end{tabular}}  \\ \hline
\multicolumn{1}{|c|}{\begin{tabular}[c]{@{}c@{}}\textbf{vehicle}\\ (846,19) \end{tabular}}    
    & \multicolumn{1}{c|}{\begin{tabular}[c]{@{}c@{}} 101\\ (0.35,0.32,0.9) \end{tabular}}  
    & \multicolumn{1}{c|}{\begin{tabular}[c]{@{}c@{}}  118\\ (0.35,0.32,0.93)  \end{tabular}}  
    & \multicolumn{1}{c|}{\begin{tabular}[c]{@{}c@{}} 90\\ (0.36,0.31,0.89)  \end{tabular}}   
    & \multicolumn{1}{c|}{\begin{tabular}[c]{@{}c@{}} 262\\ (0.36,0.33,0.93)  \end{tabular}} 
    & \multicolumn{1}{c|}{\begin{tabular}[c]{@{}c@{}}294\\ (0.34,0.32,0.94) \end{tabular}}  \\ \hline
\multicolumn{1}{|c|}{\begin{tabular}[c]{@{}c@{}}\textbf{abalone}\\ (4174,9) \end{tabular}} 
    & \multicolumn{1}{c|}{\begin{tabular}[c]{@{}c@{}}101\\ (0.36,0.31,0.87) \end{tabular}}  
    & \multicolumn{1}{c|}{\begin{tabular}[c]{@{}c@{}}117\\ (0.35,0.32,0.92) \end{tabular}}   
    & \multicolumn{1}{c|}{\begin{tabular}[c]{@{}c@{}}74\\ (0.36,0.31,0.89) \end{tabular}}  
    & \multicolumn{1}{c|}{\begin{tabular}[c]{@{}c@{}}163\\ (0.34,0.33,0.97) \end{tabular}} 
    & \multicolumn{1}{c|}{\begin{tabular}[c]{@{}c@{}}304\\ (0.34,0.32,0.94) \end{tabular}}  \\ \hline
\multicolumn{1}{|c|}{\begin{tabular}[c]{@{}c@{}}\textbf{magic}\\ (19020,11) \end{tabular}}   
    & \multicolumn{1}{c|}{\begin{tabular}[c]{@{}c@{}}7\\ (0.37,0.32,0.86)\end{tabular}}   
    & \multicolumn{1}{c|}{\begin{tabular}[c]{@{}c@{}}13\\ (0.37,0.33,0.87) \end{tabular}}  
    & \multicolumn{1}{c|}{\begin{tabular}[c]{@{}c@{}}5\\ (0.38,0.32,0.85) \end{tabular}}   
    & \multicolumn{1}{c|}{\begin{tabular}[c]{@{}c@{}}44\\ (0.37,0.34,0.91) \end{tabular}} 
    & \multicolumn{1}{c|}{\begin{tabular}[c]{@{}c@{}}16\\ (0.36,0.32,0.88) \end{tabular}}  \\ \hline
\multicolumn{1}{|c|}{\begin{tabular}[c]{@{}c@{}}\textbf{onlinenews}\\ (39644,59) \end{tabular}}   
    & \multicolumn{1}{c|}{\begin{tabular}[c]{@{}c@{}} 6585322 \\ (0.48,0.48,1) \end{tabular}}   
    & \multicolumn{1}{c|}{\begin{tabular}[c]{@{}c@{}} 6575127\\ (0.48,0.48,1)\end{tabular}}  
    & \multicolumn{1}{c|}{\begin{tabular}[c]{@{}c@{}} 6355817\\ (0.48,0.46,0.95) \end{tabular}}   
    & \multicolumn{1}{c|}{\begin{tabular}[c]{@{}c@{}}6650640\\ (0.48,0.48,1) \end{tabular}} 
    & \multicolumn{1}{c|}{\begin{tabular}[c]{@{}c@{}} 320054\\ (0.42,0.43,1) \end{tabular}}  \\ \hline
\multicolumn{1}{|c|}{\begin{tabular}[c]{@{}c@{}}\textbf{globalhousing}\\ (200000,24) \end{tabular}}    
    & \multicolumn{1}{c|}{\begin{tabular}[c]{@{}c@{}} 5424\\ (0.39,0.39,1) \end{tabular}}   
    & \multicolumn{1}{c|}{\begin{tabular}[c]{@{}c@{}} 5382\\(0.39,0.39,1) \end{tabular}}  
    & \multicolumn{1}{c|}{\begin{tabular}[c]{@{}c@{}} 5258\\ (0.40,0.38,0.96) \end{tabular}}   
    & \multicolumn{1}{c|}{\begin{tabular}[c]{@{}c@{}} 5847\\ (0.39,0.39,0.99) \end{tabular}} 
    & \multicolumn{1}{c|}{\begin{tabular}[c]{@{}c@{}} 549\\ (0.39,0.39,0.99) \end{tabular}}  \\ \hline
		\end{tabular}
	}
\end{table}

\New{Overall, the fuzzy models ($\TP$, $\TM$, $(T_{-10}^{\bm{SS}}, I^{\bm{SS}}_{k_{-10}})$, and $(\TLK, I_{0.01})$) exhibit very similar behavior in terms of the quality measures Coverage, Support, and Confidence, which generally remain stable around 0.35, 0.32, and above 0.9, respectively. The main differences lie in the number of extracted rules: $(\TLK, I_{0.01})$ tends to produce the largest rule bases, particularly in smaller datasets such as \textit{iris} and \textit{vehicle}, while maintaining high confidence values. For larger datasets, such as \textit{onlinenews} and \textit{globalhousing}, all fuzzy models scale in a similar manner, generating thousands or even millions of rules with high confidence (close to~1) and moderate coverage and support.   Regarding the crisp case, it generally produces a higher number of rules than the fuzzy approaches in small and medium-sized datasets, while preserving similar quality levels. However, for the largest datasets, the number of extracted crisp rules is considerably lower, suggesting a more restrictive or selective rule generation process.}

\New{Nonetheless, it is hard to get an idea of how different the outputs are without examining each set of output rules. Also, since all the measures introduced in Section \ref{subsec:measures} are defined in terms of the pair $(T,I)$, it is not adequate to compare any of these measures when different pairs $(T,I)$ are considered. To overcome this issue, we propose a metric for comparing two sets of fuzzy rules using the Jaccard similarity index \cite{Tan2018}. Specifically, if $SR_1 = \{R_{1,1},\dots,R_{1,k}\}$ and $SR_2=\{R_{2,1},\dots,R_{2,k}\}$ are two sets of fuzzy rules, we define the percentage of similarity as
$SR_2=\{R_{2,1},\dots,R_{2,k}\}$ are two sets of fuzzy rules, we define the percentage of similarity as
$$ Similarity(SR_1,SR_2)=\frac{|SR_1 \cap SR_2|}{|SR_1 \cup SR_2|} \cdot 100,$$
where we have considered that $R_{1,l} = R_{2,\tilde{l}}$ if they involve the same target class and the same linguistic labels of the feature variables. The resulting similarity measure can thus be applied to compare rule sets induced with different operators or fuzzy partitions, and it seamlessly accommodates the crisp case. Figure~\ref{figure:results2} reports the mean pairwise similarity among the rule sets obtained in Table~\ref{table:results1} for the models under consideration.}
\begin{figure}[!htbp]
	\centering
	\includegraphics[scale=0.4]{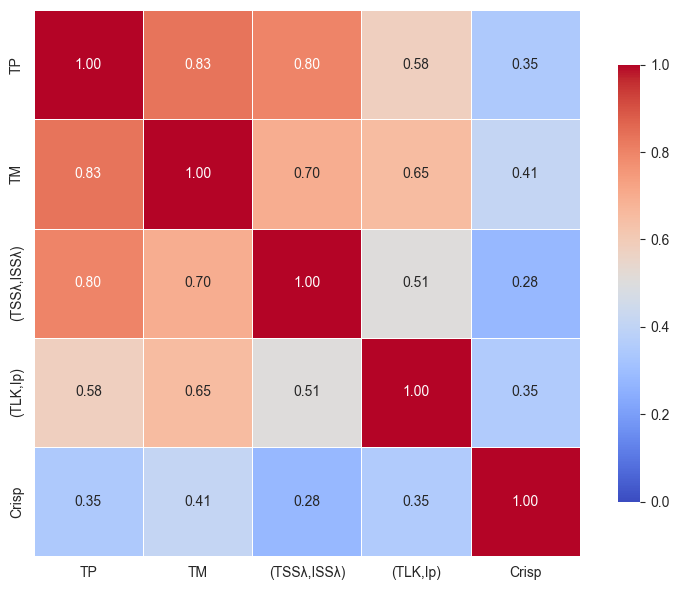}
	\caption{Mean similarity between the set of rules obtained in Table \ref{table:results1} for different models.} \label{figure:results2}
\end{figure}
\New{In Figure \ref{figure:results2} we can observe to main patterns. First, \TP, \TM, and \((T_{-10}^{\bm{SS}}, I^{\bm{SS}}_{k_{-10}})\) show the largest mutual overlaps, despite \((T_{-10}^{\bm{SS}}, I^{\bm{SS}}_{k_{-10}})\) being non-symmetric. Quantitatively, the off-diagonal similarities are high for \TP--\TM (0.83) and \TP--\((T_{-10}^{\bm{SS}}, I^{\bm{SS}}_{k_{-10}})\) (0.8), and moderate for \TM--\((T_{-10}^{\bm{SS}}, I^{\bm{SS}}_{k_{-10}})\) (0.7). This suggests that, under our thresholds, these three settings recover a largely overlapping core of rules—even when asymmetry is present. A plausible explanation is the structure of \((T_{-10}^{\bm{SS}}, I^{\bm{SS}}_{k_{-10}})\), for which the support saturates to~1 on a specific region of the domain (see Eq.~(\ref{eq:mtc_kimpl})). Second, \((\TLK, I_p)\) stands apart: its similarity with the others is only moderate (0.58 with \TP, 0.65 with \TM) and lowest against \((T_{-10}^{\bm{SS}}, I^{\bm{SS}}_{k_{-10}})\) (0.51), consistent with its explicitly implicative, directional semantics.  Crisp behaves differently for a fundamental reason: the quality metrics (coverage, support, confidence) are computed on hard sets and do \emph{not} account for fuzzy reasoning (no partial memberships or generalized modus ponens). Consequently, its rule set reflects a different inductive bias. This is visible in its low overlap with the fuzzy variants (0.35 with \TP, 0.41 with \TM, 0.28 with \((T_{-10}^{\bm{SS}}, I^{\bm{SS}}_{k_{-10}})\), and 0.35 with \((\TLK, I_p)\)); its closest match is \TM, which is more ``crisp-like'' in conjunction evaluation. Overall, Figure~\ref{figure:results2} indicates a shared rule core among \TP, \TM, and the (non-symmetric) \((T_{-10}^{\bm{SS}}, I^{\bm{SS}}_{k_{-10}})\), while the implicative \((\TLK, I_p)\) contributes distinct rule sets—retrieving information than other fuzzy models, as we discussed also in Section \ref{subsec:directional_patterns}.}


\section{Conclusions}\label{sec:conclusions}

\New{In this paper, we have revisited the concept of fuzzy implicative rules, and we have provided a solid theoretical framework for any algorithm interested in mining fuzzy rules modeled as logical conditionals. We have proved that for important characteristics like a monotone support or being the generalization of other frameworks, it is necessary to introduce a new property which we have called the monotonicity of the generalized modus ponens or \MTC. We have studied in depth this new property and we have provided different valid solutions. Although there are plenty of fuzzy implication functions available in the literature, we have shown that finding meaningful solutions remains a challenging problem, especially if the user wants to model the evaluation of a rule as a non-symmetric quantity. Also, we disclose that many pairs of operators lead to a symmetric expression of the generalized modus ponens, even if the fuzzy implication function involved is non-commutative. Further, although for some operators the support is non-symmetric, the set of rules obtained may not be much different from symmetric approaches. This shows the importance of the selection of the fuzzy operators depending on the patterns to be captured by the fuzzy rules.}

\New{Further, we have developed an open-source Python implementation of our framework for mining fuzzy implicative rules, proving that our approach yields distinct insights by introducing a novel similarity measure. As future work, we plan to further investigate the \MTC property to design more diverse fuzzy operators for modeling fuzzy implicative rules. In addition, we aim to incorporate fuzzy implicative rules into other rule-mining techniques such as subgroup discovery and exception rules. More broadly, the role of fuzzy implication functions in directional fuzzy rules deserves deeper exploration, for example within fuzzy rule-based approaches like Mamdani or TSK systems. Finally, we intend to develop heuristics and more specific optimization techniques to enhance the efficiency of mining fuzzy implicative rules.}


\bibliographystyle{IEEEtran}
\bibliography{arxiv}

\end{document}